\documentclass[runningheads]{llncs}
\usepackage{graphicx}

\usepackage{amsmath}
\usepackage{amssymb}
\usepackage{commath}
\usepackage{mathtools}
\usepackage{hyperref}
\usepackage{todonotes}
\usepackage{tikz}
\usepackage{caption}
\usepackage{subcaption}
\usepackage{thmtools} 
\usepackage{thm-restate}

\graphicspath{{graphics/}}

\newcommand{\eps}{\varepsilon}

\DeclareMathOperator{\ex}{ex}
\DeclareMathOperator{\capacity}{cap}
\DeclareMathOperator{\ch}{ch}

\begin{document}
	
	\title{Dynamic Flows with Time-Dependent Capacities}
	
	%
	%
	
	\author{Thomas Bläsius\inst{1}\orcidID{0000-0003-2450-744X} \and Adrian Feilhauer\inst{1} \and Jannik Westenfelder\inst{1}}
	
	\institute{Karlsruhe Institute of Technology (KIT), 76131 Karlsruhe, Germany\\
		\email{\{thomas.blaesius, adrian.feilhauer\}@kit.edu}}
	
	\authorrunning{T. Bl\"asius, A. Feilhauer, J. Westenfelder}
	
	
	%
	
	
	\maketitle
	\begin{abstract}
		Dynamic network flows, sometimes called flows over time, extend the
		notion of network flows to include a transit time for each edge.
		While Ford and Fulkerson showed that certain dynamic flow problems can be solved via
		a reduction to static flows, many advanced models considering
		congestion and time-dependent networks result in NP-hard problems.
		To increase understanding of these advanced dynamic flow settings we
		study the structural and computational complexity of the canonical
		extensions that have time-dependent capacities or time-dependent
		transit times.
		
		If the considered time interval is finite, we show that already a
		single edge changing capacity or transit time once makes the
		dynamic flow problem weakly NP-hard.
		In case of infinite considered time, one change in transit time or two changes in capacity make the problem weakly NP-hard.  For just one capacity change, we conjecture that the problem can be solved in polynomial time.
		Additionally, we show the structural property that dynamic cuts
		and flows can become exponentially complex in the above settings
		where the problem is NP-hard.  We further show that, despite the
		duality between cuts and flows, their complexities can be
		exponentially far apart.
	\end{abstract}
	
	\newpage 
	
	\section{Introduction}
	Network flows are a well established way to model transportation of
	goods or data through systems representable as graphs. Dynamic flows
	(sometimes called flows over time) include the temporal component by
	considering the time to traverse an edge. They were introduced by Ford
	and Fulkerson \cite{ford1958constructing}, who showed that maximum dynamic flows in static networks can be found using \emph{temporally repeated flows}, which send flow over paths of a static maximum flow as long as possible.
	
	Since capacities in real-world networks tend to be more dynamic,
	several generalizations have been considered in the literature. One
	category here is congestion modeling networks, where transit times of
	edges can depend on the flow routed over them
	\cite{kohler2005flows,langkau2003flows}. Other generalizations model
	changes in the network independently from the routed flow
	\cite{halpern1979generalized,tjandra2003dynamic,sha2000maximum}.  This
	makes it possible to model known physical changes to the network and
	allows for situations, where we have estimates of the overall
	congestion over time that is caused by external entities that are not
	part of the given flow problem.  There are also efforts to include
	different objectives for the flow, e.g., for evacuation scenarios, it
	is beneficial for a flow to maximize arrival for all times, not just
	at the end of the considered time interval \cite{disser2018simplex}.
	
	Most problems modeling congestion via flow-dependent transit times are
	NP-hard.  If the transit time depends on the current load of the edge,
	the flow problems become strongly NP-hard and no $\eps$ approximation
	exists unless $\text{P}=\text{NP}$ \cite{kohler2005flows}.  If the
	transit time of an edge instead only depends on its inflow rate while
	flow that entered the edge earlier is ignored the flow problems are
	also strongly NP-hard \cite{langkau2003flows}.
	When allowing to store flow at vertices, pseudo-polynomial algorithms
	are possible if there are time-dependent
	capacities~\cite{halpern1979generalized} and if there additionally are
	time-dependent transit times~\cite{tjandra2003dynamic,sha2000maximum}.
	In the above mentioned evacuation scenario, one aims at finding the
	so-called earliest arrival flow (EAF).  It is also NP-hard in the
	sense that it is hard to find the average arrival time of such a flow
	\cite{disser2018simplex}.  Moreover, all known algorithms to find EAFs
	have worst case exponential output size for all known encodings~\cite{disser2018simplex}.
	
	In this paper, we study natural generalizations of dynamic flows that
	have received little attention so far, allowing time-dependent
	capacities or time-dependent transit times. We prove that
	finding dynamic flows with time-dependent capacities or time-dependent
	transit times is weakly NP-hard, even if the graph is acyclic and only
	a single edge experiences a capacity change at a single point in time.
	This shows that a single change in capacity already increases the
	complexity of the -- otherwise polynomially solvable -- dynamic flow
	problem.  It also implies that the dynamic flow problem with
	time-dependent capacities is not FTP in the number of capacity
	changes.  The above results hold in the setting where the considered
	time interval is finite.  If we instead consider infinite time, the
	results remain the same for time-dependent transit times.  For
	time-dependent capacities, two capacity changes make the problem
	weakly NP-hard. We conjecture that it can be solved in polynomial when there is only one change.
	
	Beyond these results on the computational complexity, we provide
	several structural insights.  For static flows, one is usually not
	only interested in the flow value but wants to output a maximum flow
	or a minimum cut.
	The concept of flows translates more or less
	directly to the dynamic setting~\cite{ford1958constructing}, we need
	to consider time-dependent flows and cuts if we have time-dependent
	capacities or transit times.  In this case, instead of having just one
	flow value per edge, the flow is a function over time.  Similarly, in
	a dynamic cut, the assignment of vertices to one of two partitions
	changes over time.  The cut--flow duality, stating that the capacity
	of the minimum cut is the same as the value of the maximum flow also
	holds in this and many related settings
	\cite{koch2011continuous,philpott1990continuous,tjandra2003dynamic}.
	Note that the output complexity can potentially be large if the flow
	on an edge or the partition of a vertex in a cut changes often.  For
	dynamic flows on static graphs (no changes in capacities or transit
	times) vertices start in the target vertices' partition and at some
	point change to the source partition, but never the other way
	\cite{skutella2009introduction}, which shows that cuts have linear
	complexity in this setting.
	
	In case of time-dependent capacities or transit times, we show that
	flow and cut complexity are sometimes required to be exponential.
	Specifically, for all cases where we show weak NP-hardness, we also
	give instances for which every maximum flow and minimum cut have
	exponential complexity.  Thus, even a single edge changing capacity or
	transit time once can jump the output complexity from linear to
	exponential.  Moreover, we give examples where the flow complexity is
	exponential while there exists a cut of low complexity and vice versa.
	
	We note that the scenario of time-dependent capacities has been
	claimed to be strongly NP-complete~\cite{sha2000maximum} before.
	However, we suspect the proof to be flawed as one can see that this
	scenario can be solved in pseudo-polynomial time.  Moreover, the above
	mentioned results on the solution complexity make it unclear whether
	the problem is actually in NP. In Appendix~\ref{sec:strong-np-compl},
	we point out the place where we believe the proof for strong
	NP-hardness is flawed.
	
	\section{Preliminaries}
	
	We consider dynamic networks $G=(V,E)$ with directed edges and
	designated source and target vertices $s,t\in V$. Edges $e=(v,w)\in E$
	have a time-dependent non negative \emph{capacity}
	$u_e \colon [0,T]\to\mathbb{R}^+_0$, specifying how much flow can
	enter $e$ via $v$ at each time.  We allow $u_e$ to be non-continuous
	but only for finitely many points in time.  In addition, each edge
	$e=(v,w)$ also has a non negative \emph{transit time}
	$\tau_e\in\mathbb{R}^+$, denoting how much time flow takes to move
	from $v$ to $w$ when traversing $e$.  Note that the capacity is
	defined on $[0, T]$, i.e., time is considered from $0$ up to a
	\emph{time horizon} $T$.
	
	Let $f$ be a collection of measurable functions
	$f_e \colon [0, T - \tau_e] \to \mathbb R$, one for each edge
	$e \in E$, assigning every edge a flow value depending on the time.
	The restriction to the interval $[0, T - \tau_e]$ has the
	interpretation that no flow may be sent before time $0$ and no flow
	should arrive after time $T$ in a valid flow.  To simplify notation,
	we allow time values beyond $[0, T - \tau_e]$ and implicitly assume
	$f_e(\Theta) = 0$ for $\Theta\notin[0, T - \tau_e]$.  We call $f$ a
	\emph{dynamic flow} if it satisfies the \emph{capacity constraints}
	$f_e(\Theta) \le u_e(\Theta)$ for all $e \in E$ and
	$\Theta \in [0, T - \tau_e]$, and \emph{strong flow conservation}, which we
	define in the following.
	
	The \emph{excess flow} $\ex_f(v,\Theta)$ of a vertex
	$v$ at time $\Theta$ is the difference between flow sent to $v$ and
	the flow sent from $v$ up to time $\Theta$,
	i.e.,
	\begin{equation*}
	\ex_f(v,\Theta)\coloneqq\int_{0}^{\Theta}\sum_{e = (u, v) \in E}f_e(\zeta -
	\tau_e)
	- \sum_{e = (v, u) \in E}f_e(\zeta) \dif\zeta.
	\end{equation*}
	We have strong flow conservation if $\ex_f(v, \Theta) = 0$
	for all $v \in V\setminus\{s,t\}$ and $\Theta\in[0,T]$.
	
	The \emph{value} of $f$ is defined as the excess of the target vertex
	at the time horizon $|f|\coloneqq \ex_f(t,T)=-\ex_f(s,T)$.
	The \emph{maximum dynamic flow problem with time-dependent capacities} is
	to find a flow of maximum value.  We refer to its input as
	\emph{dynamic flow network}.
	
	A cut-flow duality similar to the one of the static maximum flow
	problem holds for the maximum dynamic flow problem with the following
	cut definition. A \emph{dynamic cut} or \emph{cut over time} is a
	partition of the vertices $(S,V\setminus S)$ for each point in time,
	where the source vertex $s$ always belongs to $S$ while the target $t$
	never belongs to $S$. Formally, each vertex $v\in V$ has a boolean
	function $S_v\colon [0,T]\to\{0,1\}$ assigning $v$ to $S$ at time
	$\Theta$ if $S_v(\Theta) = 1$.  As for the flow, we extend $S_v$
	beyond $[0, T]$ and set $S_v(\Theta) = 1$ for $\Theta>T$ for all
	$v\in V$ (including $t$). The \emph{capacity} $\capacity(S)$ of a
	dynamic cut $S$ is the maximum flow that could be sent on edges from
	$S$ to $V\setminus S$ during the considered time interval $[0,T]$,
	i.e.,
	\begin{equation*}
	\capacity(S)=\int_{0}^{T}\sum_{\substack{(v,w)\in E\\S_v(\Theta)=1\\ S_w(\Theta+\tau_e)=0}}u_{(v,w)}(\Theta)\dif\Theta.
	\end{equation*}
	An edge $(v,w)$ \emph{contributes} to the cut $S$ at time $\Theta$ if it contributes to the above sum, so $S_v(\Theta)=1\wedge S_w(\Theta+\tau_e)=0$.
	Note that this is similar to the static case, but in the dynamic
	variant the delay of the transit time needs to be considered.  Thus,
	for the edge $e = (v, w)$ we consider $v$ at time $\Theta$ and $w$ at
	time $\Theta + \tau_e$.  Moreover, setting $S_v(\Theta) = 1$ for all
	vertices $v$ if $\Theta > T$ makes sure that no point in time beyond
	the time horizon contributes to $\capacity(S)$.
	
	\begin{theorem}[Min-Cut Max-Flow Theorem \cite{philpott1990continuous,tjandra2003dynamic}]
		\label{trm:cut_flow_duality}
		For a maximum flow over time $f$ and a minimum cut over time $S$ it holds $|f|=cap(S)$.
	\end{theorem}
	\begin{proof}
		The theorem by Philpott
		\cite[Theorem~1]{philpott1990continuous} is more
		general than the setting considered here.  They in
		particular allow for time-dependent storage
		capacities of vertices.  We obtain the here stated
		theorem by simply setting them to constant zero.
		The theorem by Tjandra \cite[Theorem 3.4]{tjandra2003dynamic} is even more general and thus also covers the setting with time-dependent transit times.
		\qed
	\end{proof}
	
	Though the general definition allows the capacity functions to be
	arbitrary, for our constructions it suffices to use piecewise constant
	capacities.  We note that in this case, there always exists a maximum
	flow that is also piecewise constant, assigning flow values to a set
	of intervals of non-zero measure.  The property that the intervals
	have non-zero measure lets us consider an individual point $\Theta$ in
	time and talk about the contribution of an edge to a cut or flow at time
	$\Theta$, as $\Theta$ is guaranteed to be part of a non-empty interval
	with the same cut or flow.  For the remainder of this paper, we assume
	that all flows have the above property.
	
	We define the following additional useful notation.  We use
	$S(\Theta)\coloneqq\{v\in V\mid S_v(\Theta)=1\}$ and
	$\bar{S}(\Theta)\coloneqq\{v\in V\mid S_v(\Theta)=0\}$ to denote the cut at time
	$\Theta$.  Moreover, a vertex $v$ changes its partition at time
	$\Theta$ if $S_v(\Theta - \eps) \neq S_v(\Theta + \eps)$ for every
	sufficiently small $\eps > 0$.  We denote a change from $S$ to
	$\bar{S}$ with $S_v\xrightarrow{\Theta}\bar{S}_v$ and a change in the other \emph{direction} from
	$\bar{S}$ to $S$ with $\bar{S}_v\xrightarrow{\Theta}S_v$.
	We denote the number of partition changes of a vertex $v$ in a cut $S$
	with $\ch_v(S)$.  Moreover the total number of changes in $S$ is the
	\emph{complexity} of the cut $S$.
	For a flow $f$, we define changes on edges as well as the complexity of $f$ analogously.
	
	In the above definition of the maximum dynamic flow problem we allow
	time-dependent capacities but assume constant transit times.  Most of
	our results translate to the complementary scenario where transit
	times are time-dependent while capacities are constant. In this setting $\tau_e(\Theta)$ denotes how much time flow takes to traverse $e$, if it enters at time $\Theta$. Similarly to the above definition, we allow $\tau_e$ to be non-continuous for finitely many points in time.
	
	Additionally we look at the scenarios where infinite time ($\Theta\in(-\infty,\infty)$) is
	considered instead of only considering times in $[0,T]$.  This removes
	structural effects caused by the boundaries of the considered time
	interval.  Intuitively, because we are working with piecewise constant functions with finitely many incontinuities, there exists a point in time $\Theta$ that is
	sufficiently late that all effects of capacity changes no longer play
	a role.  From that time on, one can assume the maximum flow and
	minimum cut to be constant.  The same holds true for a sufficiently
	early point in time.  Thus, to compare flow values it suffices to look
	at a finite interval $I$.  Formally, $f$ is a maximum dynamic flow
	with \emph{infinite considered time} if it is constant outside of $I$
	and maximum on $I$, such that for any larger interval $J\supset I$
	there exists a large enough interval $K\supset J$ so that a maximum
	flow with considered time interval $K$ can be $f$ during $J$.  Minimum
	cuts with infinite considered time are defined analogously.  Such
	maximum flows and minimum cuts always exist as temporally repeated
	flows provide optimal solutions to dynamic flows and we only allow
	finitely many changes to capacity or traversal time.
	
	We will need the set of all integers up to $k$ and denote it $[k]\coloneqq\{i\in\mathbb{N}^+|i\leq k\}$.
	
	\section{Computational Complexity}
	\label{sec:comp-compl}
	
	In this section we study the computational complexity of the dynamic
	flow problem with time-dependent capacities or transit times.  We
	consider finite and infinite time. For all cases except for a single capacity change with infinite considered time, we prove NP-hardness.
	
	
	We start by showing hardness in the setting where we have
	time-dependent capacities with only one edge changing capacity once.
	Our construction directly translates to the setting of infinite
	considered time with one edge changing capacity twice. 
	For the case of time-dependent transit times we prove hardness for one
	change even in the infinite considered time setting.  This also
	implies hardness for one change when we have a finite time horizon.
	
	We reduce from the \emph{partition problem}, which is defined as
	follows.  Given a set of positive integers $S=\{b_1,\dots,b_k\}$ with
	$\sum_{i=1}^k b_i=2L$, is there a subset $S'\subset S$ such that
	$\sum_{a\in S'}a=L$?
	
	\begin{theorem}
		\label{trm:weak_NP_hardness}
		The dynamic flow problem with time-dependent capacities is weakly NP-hard, even for acyclic graphs with only one capacity change.
	\end{theorem}
	\begin{proof}
		Given an instance of the partition problem, we construct $G=(V,E)$
		as shown in
		Figure~\ref{fig:NP_reduction_timeDependentCapacitiesSingleChange} and show that a solution to partition is equivalent to a flow of value $1$ in $G$.
		
		Every $b_i\in S$ corresponds to a vertex $x_i$ which can be reached
		by $x_{i-1}$ with one edge of transit time $b_i$ and one bypass edge of
		transit time zero. The last of these vertices $x_k$ is connected to the
		target $t$ with an edge only allowing flow to pass during
		$[L+1,L+2]$, where the lower border is ensured by the capacity
		change of $(x_k,t)$ and the upper border is given by the time
		horizon $T=L+3$. The source $s$ is connected to $x_0$ with an edge
		of low capacity $\frac{1}{L+1}$, so that the single flow unit that
		can enter this edge in $[0,T-2]$ can pass $(x_k,t)$ during one time unit.
		
		Since a solution to the partition problem is equivalent to a path of
		transit time $L$ through the $x_i$, we additionally provide paths of
		transit time $0,1,\dots,L-1$ bypassing the $b_i$ edges via the $y_i$ so
		that a solution for partition exists, if and only if flow of value
		$1$ can reach $t$.
		To provide the bypass paths, we set $\ell\in\mathbb{N}_0$ so that $L=2^{\ell+1}+r,\ r\in\mathbb{N}_0,r<2^{\ell+1}$ and define vertices 
		$y_i, i\in\mathbb{N}_0, i\leq\ell$. They create a path of transit time
		$L-1$ where the edges' transit times are powers of two and one edge of
		transit time $r$ and all edges can be bypassed by an edge with transit time zero.
		This allows all integer transit times smaller than $L-1$. All edges
		except for $(s,x_0)$ have unit capacity when they are active.
		
		Given a solution $S'$ to the partition problem, we can route flow leaving $s$ during $[0,1]$ through the $x_i$ along the non zero transit time edges if and only if the corresponding $b_i$ is in $S'$. Flow leaving $s$ in $[1,L+1]$ can trivially reach $x_k$ during $[L+1,L+2]$ using the bypass paths, providing a maximum flow of $1$.
		
		Only one unit of flow can reach $x_0$ until $L+2$, considering the time horizon $T=L+3$ and the transit time of $(x_k,t)$, the flow can have value at most $1$. Given a flow that sends one unit of flow to $t$, we can see that the flow has to route all flow that can pass $(s,x_0)$ during $[0,L+1]$ to $t$. Due to the integrality of  transit times, the flow leaving $s$ during $[0,1]$ has to take exactly time $L$ to traverse from $x_0$ to $x_k$. The bypass paths via $y_0$ are too short for this. As such, this time is the sum of edge transit times taken from the partition instance and zeroes from bypass edges, and there exists a solution $S'$ to the partition problem that consists of the elements corresponding to the non zero transit time edges taken by this flow.
		\qed
	\end{proof}
	
	\begin{figure}
		\centering
		\includegraphics{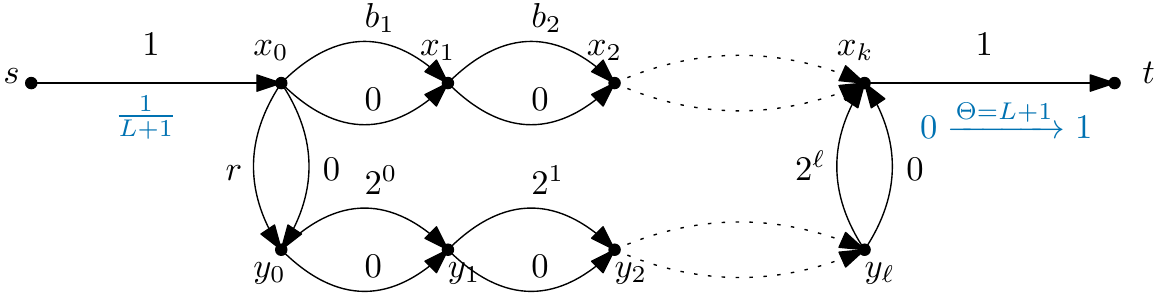}
		\caption{Graph constructed for the reduction of the partition problem to dynamic flow with time-dependent capacities. Flow leaving $s$ at time zero can only reach $t$ if it takes exactly time $L$ to traverse from $x_0$ to $x_k$, such choosing a partition. Black numbers are transit times, blue numbers indicate capacity, all unspecified capacities are $1$, time horizon is $T=L+3$.}
		\label{fig:NP_reduction_timeDependentCapacitiesSingleChange}
	\end{figure}
	
	\begin{restatable}{corollary}{weakNpHardnessCapacityInfTime}
		The dynamic flow problem with time-dependent capacities and infinite considered time is weakly
		NP-hard, even for acyclic graphs with only two capacity changes.
	\end{restatable}
	\begin{proof}
		In the proof of Theorem~\ref{trm:weak_NP_hardness}, we restricted
		the flow on the edge from $x_k$ to $t$ to have non-zero capacity
		only at time $[L + 1, L + 2]$.  For
		Theorem~\ref{trm:weak_NP_hardness}, we achieved the lower bound with
		one capacity change and the upper bound with the time horizon.
		Here, we can use the same construction but use a second capacity
		change for the upper bound.
		\qed
	\end{proof}
	
	For the case of time-dependent transit times, we use a similar
	reduction.  We start with the case of infinite considered time.
	
	\begin{restatable}{theorem}{weakNPHardnessInfTimeTimeDependentTransitTimes}
		\label{trm:weak_NP_hardness_infTime_timeDependentTransitTimes}
		The dynamic flow problem with infinite considered time and
		time-dependent transit times is weakly NP-hard, even for acyclic graphs with only one
		transit time change.
	\end{restatable}
	\begin{proof}
		Similar to the proof of Theorem~\ref{trm:weak_NP_hardness} we give a reduction of the partition problem. The constructed graph can be seen in Figure~\ref{fig:NP_reduction_timeDependentTransitTimeSingleChangeInfiniteTime}. We want to link the existence of a transit time $L$ path to a maximum flow sending $1$ flow per time from $s$.
		For this, we start the graph  with an edge $(s,x_0)$ whose transit time gets reduced from $1$ to zero at time $\Theta=0$. This results in $2$ units of flow reaching $x_0$ at time $\Theta=0$, while only a flow of $1$ can traverse $(x_0,t)$. This means that flow of $1$ has to pass through the $x_i$ and $y_i$.
		The paths through the $x_i$ and the bypass paths through the $y_i$ function like in the proof of Theorem~\ref{trm:weak_NP_hardness}, but here the bypass edges also provide paths of transit times $L+1$ to $2L$.
		Because the edge $(x_k,t)$ has capacity $u_{(x_k,t)}=\frac{1}{2L+1}$, the flow routed through $x_k$ has to arrive at $x_k$ using at least $2L+1$ paths with different transit times.
		The paths from $x_0$ to $x_k$ have integer transit times between zero and $2L$, so to route the extra unit of flow arriving at $x_0$ at time $\Theta=0$, flow needs to be routed through one path of each integer transit time between $0$ and $2L$. The bypass paths do not offer a path of transit time $L$, so, like in the proof of Theorem~\ref{trm:weak_NP_hardness}, this flow unit can be completely routed through the network if and only if the partition problem has a solution. 
		\qed
	\end{proof}

	\begin{figure}
		\centering
		\includegraphics{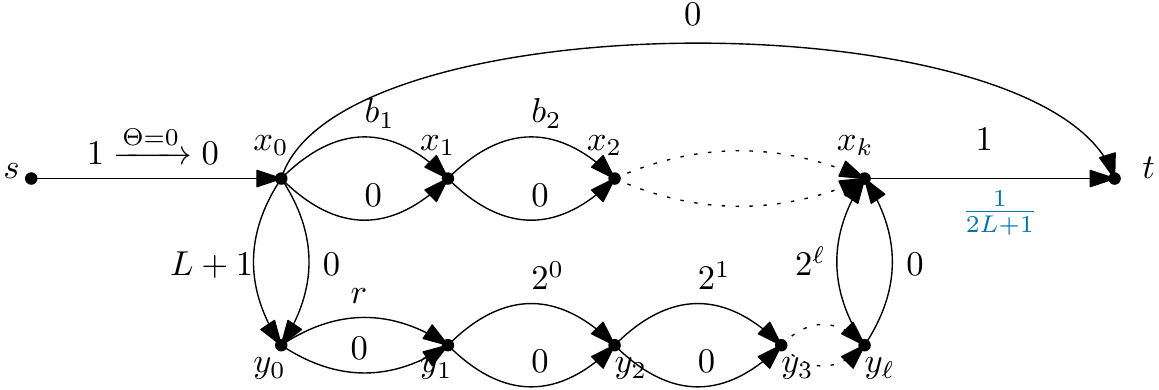}
		\caption{Graph constructed for the reduction of the partition problem to dynamic flow with infinite considered time and time-dependent transit times. The flow units leaving $s$ at times $-1$ and zero can only reach $t$ if they take the direct $(x_0,t)$ edge and $2L+1$ paths with different transit times to $x_k$.
		Black numbers are transit times, blue numbers specify capacity, all unspecified capacities are $1$.}
		\label{fig:NP_reduction_timeDependentTransitTimeSingleChangeInfiniteTime}
	\end{figure}
	
	To translate this result to the case of a finite time horizon, note
	that we can use the above construction and choose the time horizon
	sufficiently large to obtain the following corollary.
	
	\begin{restatable}{corollary}{weakNPHardnessTransitTimeFiniteTime}
		The dynamic flow problem with time-dependent transit times is weakly NP-hard, even for acyclic graphs with only one transit time change.
	\end{restatable}
	\begin{proof}
		Using the construction of the proof for Theorem~\ref{trm:weak_NP_hardness_infTime_timeDependentTransitTimes}, we can restrict time to the interval $[-1,2L+1]$, then a solution to the partition problem is equivalent to the existence of a flow of value $2L+2$. To get a considered time interval from zero to a time horizon, we let the transit time change of $(s,x_0)$ occur at time $\Theta=1$ instead and set $T=2L+2$.
		\qed
	\end{proof}
	
	This leaves one remaining case: infinite considered time and a single capacity change.
	For this case, we can show that there always exists a minimum dynamic cut, where each vertex changes partition at most once and all partition changes are of the same direction. Furthermore, for given partitions before and after the changes, a linear program can be used to find the optimal transition as long as no vertex changes partition more than once and all partition changes are of the same direction.
	This motivates the following conjecture.
	
	\begin{conjecture}
		\label{cjt:infTimeSingleCapacityChange_poly}
		The minimum cut problem in a dynamic flow network with only a single change in capacity and infinite considered time can be solved in polynomial time.
	\end{conjecture}

	\section{The Complexity of Maximum Flows and Minimum Cuts}
	\label{sec:compl-maxim-flows}
	
	We first construct a dynamic flow network such that
	all maximum flows and minimum cuts have exponential complexity.
	Afterwards, we show that there are also instances that require
	exponentially complex flows but allow for cuts of linear size and vice
	versa.
	These results are initially proven for a single change in capacity and are then shown to also hold in the setting with time dependent transit times, likewise with only one change in transit time required.
	
	\subsection{Exponentially Complex Flows and Cuts}
	\label{subsec:exponentially_complex_flows_and_cuts}
	
	We initially focus on the complexity of cuts and only later show that
	it transfers to flows.  Before we start the construction, note that
	the example in Figure~\ref{fig:forcedChangeBasicIdea} shows how the
	partition change of two vertices $a$ and $b$ can force a single vertex
	$v$ to change its partition back and forth.  This type of enforced
	partition change of $v$ is at the core of our construction.
	
	\begin{figure}
		\centering
		\includegraphics{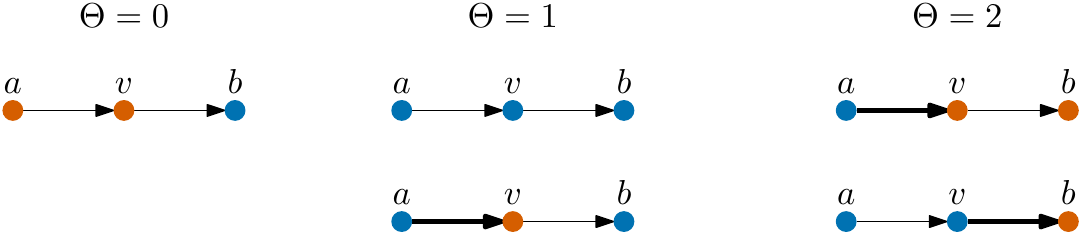}
		\caption{Example where $a$ changes from $\bar{S}$ (red) to $S$
			(blue) at time $1$ and $b$ changes from $S$ to $\bar{S}$ at time
			$2$.  Only edges from $S$ to $\bar{S}$ contribute to the cut (bold
			edges).  Assuming $v$ starts in $\bar{S}$ and
			$u_{(a, v)} < u_{(v, b)}$ as well as $\tau_{(a,v)}=\tau_{(v,b)}=0$, $v$ has to change to $S$ at time $1$
			and back to $\bar{S}$ at time $2$ in a minimum cut (top row).  The
			bottom row illustrates the alternative (more expensive) behavior
			of $v$.}
		\label{fig:forcedChangeBasicIdea}
	\end{figure}
	
	More specifically, we first give a structure with which we
	can force vertices to mimic the partition changes of other vertices,
	potentially with fixed time delay. 
	
	The \emph{mimicking gadget} links two non terminal vertices
	$a,b\in V\setminus\{s,t\}$ using edges $(a,b),(b,t)\in E$ with
	capacities $u_{(a,b)}=\alpha,u_{(b,t)}=\beta$.  The following lemma
	shows what properties $\alpha$ and $\beta$ need to have such that the
	mimicking gadget does its name credit, i.e., that $b$ mimics $a$ with
	delay $\tau_{(a,b)}$.  A visualization of the mimicking gadget is
	shown in Figure~\ref{fig:partitionMimickingGadget}.
	
	\begin{figure}
		\centering
		\includegraphics{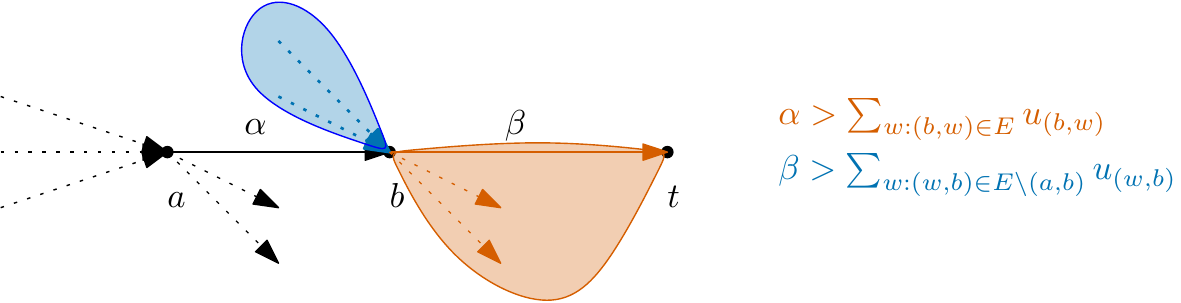}
		\caption{Gadget linking the partitions of two vertices $a$ and $b$,
			so that $b$ mimics $a$ with a delay of $\tau_{(a,b)}$;
			$\alpha,\beta$ are capacities.}
		\label{fig:partitionMimickingGadget}
	\end{figure}
	
	\begin{lemma}
		\label{lem:mimicking_gadget}
		Let $G$ be a graph that contains the mimicking gadget as a
		sub-graph, such that
		\begin{equation*}
		\alpha > \sum_{w\mid(b,w)\in E}u_{(b,w)} \quad \text{and} \quad
		\beta > \sum_{w\mid(w,b)\in E\setminus(a,b)}u_{(w,b)}.
		\end{equation*}
		Then, $S_b(\Theta)=S_a(\Theta-\tau_{(a,b)})$ for every minimum cut
		$S$ and times $\Theta\in(\tau_{(a,b)},T-\tau_{(b,t)})$.
	\end{lemma}
	\begin{proof}
		We first show $a\in S(\Theta-\tau_{(a,b)})\implies b\in
		S(\Theta)$.
		With the partition of $a$ fixed,
		we look at
		possible contribution to $S$ of edges incident to $b$ at time $\Theta$. For $b\in\bar{S}(\Theta)$ the contribution is
		at least $\alpha$, because $\Theta\in(\tau_{(a,b)},T-\tau_{(b,t)})$
		ensures that $(a,b)$ can contribute to $S$.
		For $b\in S(\Theta)$ the contribution is at most $\sum_{w\mid(b,w)\in E}u_{(b,w)}<\alpha$. Because $S$ is a minimum cut, we obtain $b\in S(\Theta)$.
		The other direction
		$a\in\bar{S}(\Theta-\tau_{(a,b)})\implies b\in\bar{S}(\Theta)$ holds
		for similar reasons. For $b\in S(\Theta)$ the contribution is at
		least $\beta$. For $b\in\bar{S}(\Theta)$ the contribution is at most
		$\sum_{w\mid(w,b)\in E}u_{(w,b)}<\beta$.
		\qed
	\end{proof}
	
	Note that Lemma~\ref{lem:mimicking_gadget} does not restrict the edges
	incident to $a$.  Thus, we can use it rather flexibly to transfer
	partition changes from one vertex to another.
	
	To enforce exponentially many partition changes, we next give a gadget
	that can double the number of partition changes of one vertex.  To
	this end, we assume that, for every integer $i\in[k]$,
	we already have access to vertices $a_i$ with \emph{period}
	$p_i \coloneqq 2^i$, i.e., $a_i$ changes partition every $p_i$ units of time.
	Note that $a_1$ is the vertex with the most changes.  With this, we
	construct the so-called binary counting gadget that produces a vertex
	$v$ with period $p_0 = 1$, which results in it having twice as many changes as
	$a_1$.  Roughly speaking, the binary counting gadget, shown in
	Figure~\ref{fig:binaryCountingGadget}, consists of the above mentioned
	vertices $a_i$ together with additional vertices $b_i$ such that $b_i$
	mimics $a_i$.  Between the $a_i$ and $b_i$ lies the central vertex $v$
	with edges from the vertices $a_i$ and edges to the $b_i$.
	Carefully chosen capacities and synchronization between the $a_i$ and $b_i$ results in $v$ changing partition every step.
	
	To iterate this process using $v$ as vertex for the binary counting
	gadget of the next level, we need to ensure functionality with the
	additionally attached edges of the mimicking gadget.
	
	The \emph{binary counting gadget} $H_k$ shown in
	Figure~\ref{fig:binaryCountingGadget} is formally defined as follows.
	It contains the above mentioned vertices $a_i,b_i$ for $i\in[k]$ and the vertex $v$.  Additionally, it contains
	the source $s$ and target $t$.  On this vertex set, we have five types
	of edges.  All of them have transit time $1$ unless explicitly specified
	otherwise.  The first two types are the edges $(a_i, b_i)$ and
	$(b_i, t)$ for $i \in [k]$, which form a mimicking gadget.  We set
	$\tau_{(a_i, b_i)} = p_i + 1$ which makes $b_i$ mimic the changes of
	$a_i$ with delay $p_i + 1$.  Moreover, we set
	$u_{(a_i, b_i)} = \alpha_i \coloneqq 2^{i-1} + 2 \eps$ and
	$u_{(b_i, t)} = \beta_i \coloneqq 2^{i - 1} + \eps$.  We will see that these
	$\alpha_i$ and $\beta_i$ satisfy the requirements of the mimicking
	gadget in Lemma~\ref{lem:mimicking_gadget}.  The third and fourth
	types of edges are $(a_i, v)$ and $(v, b_i)$ for $i \in [k]$ with
	capacities $u_{(a_i, v)} = u_{(v, b_i)} = 2^{i - 1}$.  These edges
	have the purpose to force the partition changes of $v$, similar to the
	simple example in Figure~\ref{fig:forcedChangeBasicIdea}.  Finally, we
	have the edge $(s, v)$ with capacity $u_{(s, v)} = 1 - \eps$.  It has
	the purpose to fix the initial partition of $v$ and introduce some
	asymmetry to ensure functionality even if additional edges are
	attached to $v$.
	
	Our plan is to prove that the binary counting gadget $H_k$ works as
	desired by induction over $k$.  We start by defining the desired
	properties that will serve as induction hypothesis.
	
	\begin{definition}
		Let $G$ be a graph.  We say that $H_k$ is a \emph{valid binary counting
			gadget} in $G$ if $H_k$ is a subgraph of $G$ and every minimum cut
		$S$ has the following properties.
		\begin{itemize}
			\item For $i\in[k-1]$, the vertex $a_i$ has period $p_i$.  It changes
			its partition $2^{k - i}$ times starting with a change from $S$ to
			$\bar{S}$ at time $0$ and ending with a change at time $2^k-2^i$.
			\item For $i = k$, $a_k$ changes from $S$ to $\bar{S}$ at time $0$
			and additionally back to $S$ at time $2^k$.
		\end{itemize}
	\end{definition}
	
	\begin{figure}
		\centering
		\includegraphics{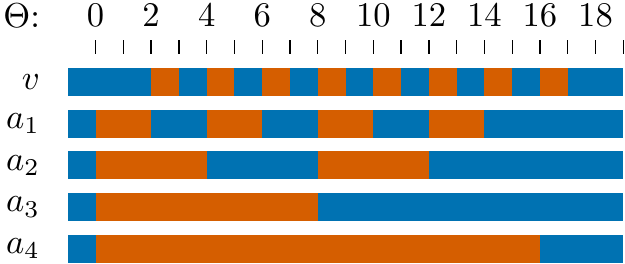}
		\caption{Visualization of the partition change patterns of a valid binary counting gadget $H_4$. In blue sections the vertex is in $S$ and in red sections it is in $\bar{S}$.}
		\label{fig:binaryCounting}
	\end{figure}
	
	Note that in a valid binary counting gadget the $a_i$ and $v$ form a binary counter from $0$ to $2^k-1$
	when regarding $\bar{S}$ as zero and $S$ as $1$, with $v$ being the
	least significant bit (shifted back two time steps); see
	Figure~\ref{fig:binaryCounting}.
	
	\begin{figure}
		\centering
		\begin{subfigure}{\textwidth}
			\centering
			\includegraphics{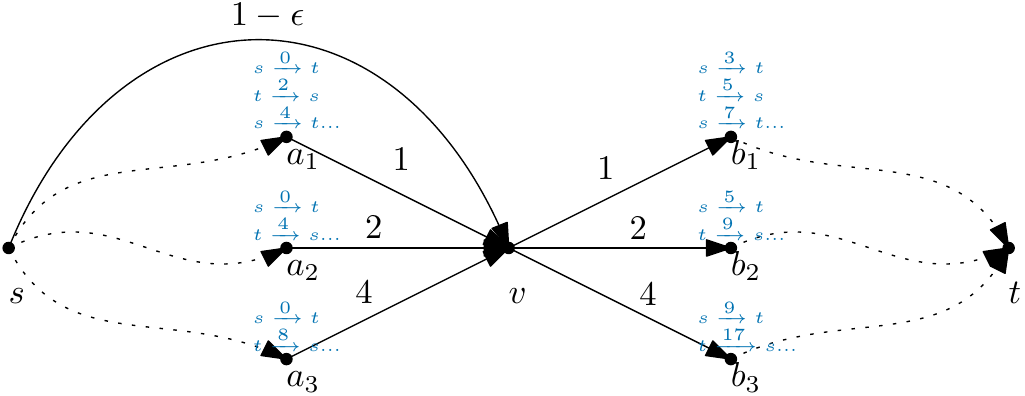}
			\caption{More intuitive visualization exemplary showing the partition changes of $a_i,b_i$ in blue, omitting the mimicking gadgets. For improved readability the partitions are denoted by their terminal, i.e. $s$ for $S$ and $t$ for $\bar{S}$.}
			\vspace{18pt}
		\end{subfigure}
		\begin{subfigure}{\textwidth}
			\centering
			\includegraphics{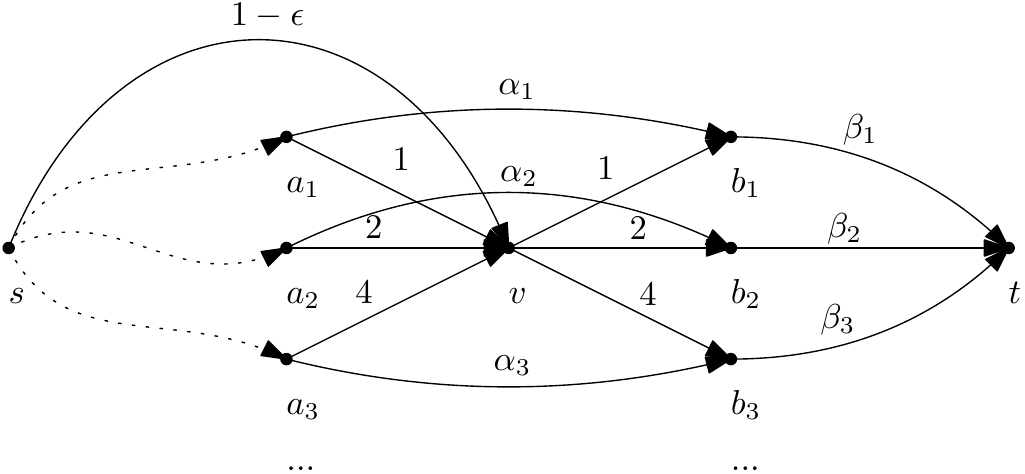}
			\caption{Full visualization with all edges, including the mimicking gadgets.}
		\end{subfigure}
		\caption{Visualization of the binary counting gadget allowing to double the number of partition changes of a single vertex, ensuring $2^k$ changes of vertex $v$ assuming that vertices $a_i,i<k$ are changing partition $\ch_{a_i}=2^{k-i}$ times each, with $\ch_{a_k}=2$ and correct timing; Black numbers are capacities.}
		\label{fig:binaryCountingGadget}
	\end{figure}
	
	\begin{restatable}{lemma}{lemBinaryCountingGadget}
		\label{lem:binary_counting_gadget}
		Let $G$ be a graph containing the valid binary counting gadget $H_k$ such
		that the central vertex $v$ has no additional incoming edges, the
		sum of the capacities of additional outgoing edges of $v$ is less than
		$1-\eps$, and no additional edges are incident to the $b_i$.  Then,
		for every minimum cut $S$, the central vertex $v$ has period
		$p_0 = 1$ and changes $2^k$ times, starting with a change from $S$
		to $\bar{S}$ at time $2$ and ending with a change at time $2^k+1$.
	\end{restatable}
	\begin{proof}
		First note that the mimicking gadget allows causing an inverted counter behavior for the $b_i$, affecting $v$ one time step before the corresponding $a_i$.
		For $\Theta\in\{0,\dots,2^k-1\}$, the $a_i$ change
		\begin{equation*}
		S_{a_i}\xrightarrow{\Theta\equiv0\bmod2^{i+1}}\bar{S}_{a_i}\\
		\bar{S}_{a_i}\xrightarrow{\Theta\equiv2^i\bmod2^{i+1}}S_{a_i}\\
		\bar{S}_{a_k}\xrightarrow{2^k}S_{a_k}.
		\end{equation*}
		The delay of $p_i+1$ for inverting the counter is chosen because the activated states of the $a_i$ and $b_i$ are opposite, i.e. $S_{a_i}(\Theta-\tau_{(a_i,v)})=1$ allows contribution of $u_{(a_i,v)}$, but $S_{b_i}(\Theta+\tau_{(v,b_i)})=0$ allows contribution of $u_{(v,b_i)}$. The reset step at time $1$, changing all $b_i$ to $S$ can be omitted, as all vertices start in $S$. So the necessary delay between $a_i$ and $b_i$ is $2^i+1$.
		The desired change pattern therefore is
		\begin{equation*}
		\bar{S}_{b_i}\xrightarrow{\Theta\equiv1\bmod2^{i+1}}S_{b_i}\\
		\bar{S}_{b_k}\xrightarrow{2^{k+1}+1}S_{b_k}\\
		S_{b_i}\xrightarrow{\Theta\equiv2^i+1\bmod2^{i+1}}\bar{S}_{b_i}
		\end{equation*}
		for $\Theta\in\{3,\dots,2^k+1\}$. This pattern is achieved by the functionality of the mimicking gadget shown in Lemma~\ref{lem:mimicking_gadget} and the fact that $G$ cannot have additional edges incident to any $b_i$.
		
		To realize that the partition changes of $v$ have to occur in the claimed way for any minimum cut $S$, we look at the edges incident to $v$. All other edges' contribution to any cut is already fixed. An edge $(a_i,v)$ contributes $2^i$ to $\capacity(S)$ if and only if $a_i\in S(\Theta-1)$ and $v\in\bar{S}(\Theta)$ ($s\in S(\Theta)$ always holds, so $(s,v)$ contributes $1-\varepsilon$ if $v\in\bar{S}(\Theta)$) for some time $\Theta$. Likewise the only way for $(v,b_i)$ to contribute to $\capacity(S)$ is $v\in S(\Theta)$ and $b_i\in\bar{S}(\Theta+1)$ for some time $\Theta$. Evaluating these contributions for the given partition changes of $a_i,b_i$ we see that the counter of the $b_i$ contributions is half a counting step -- which corresponds to one time unit -- ahead of the $a_i$ contribution counter. For the last change of $a_k$, affecting $v$ at time $2^k+1$, the $a_i$ counter gets larger than the $b_i$ counter instead of equaling it. Formally for $\Theta\in[1,2^{k}+1)$ the contribution of $(a_i,v)$ and $(v,b_i)$ edges is
		\begin{equation*}
		\sum_{i\mid S_{a_i}(\Theta-1)=1}u_{(a_i,v)}=\lfloor\frac{\Theta-1}{2}\rfloor\\
		\sum_{i\mid S_{b_i}(\Theta+1)=0}u_{(v,b_i)}=\lfloor\frac{\Theta}{2}\rfloor
		\end{equation*}
		and for $\Theta\in[2^k+1,2^k+2)$
		\begin{equation*}
		\sum_{i\mid S_{a_i}(\Theta-1)=1}u_{(a_i,v)}=2^{k}-1>2^{k-1}=\sum_{i\mid S_{b_i}(\Theta+1)=0}u_{(v,b_i)}.
		\end{equation*}
		With the additional edge $(s,v)$, the side with more potential to contribute to the cut changes every $\Theta\in\{2,\dots,2^k+1\}$, forcing $v$ to change its partition every time to ensure minimality of the cut. So the change pattern of $v$ is
		\begin{equation*}
		S_v\xrightarrow{\Theta\equiv0\bmod2}\bar{S}_v\\
		\bar{S}_v\xrightarrow{\Theta\equiv1\bmod2}S_v
		\end{equation*}
		for $\Theta\in\{2,\dots,2^k+1\}$.
		
		Edges leaving $v$ can only contribute to $\capacity(S)$ if $v\in S(\Theta)$, so whenever the gadget already ensures $v\in\bar{S}(\Theta)$, added outgoing edges cannot impede the gadgets behavior. When $H_k$ ensures $v\in S(\Theta)$, we have
		\begin{equation*}
		\sum_{i\mid S_{a_i}(\Theta-1)=1}u_{(a_i,v)}\geq\sum_{i\mid S_{b_i}(\Theta+1)=0}u_{(v,b_i)}.
		\end{equation*}
		To minimize $\capacity(S)$, the assignment $v\in S(\Theta)$ remains necessary to minimize $\capacity(S)$, because of the edge $(s,v_k)$ with capacity $1-\varepsilon$, which is larger than the sum over the capacities of all added edges.
		\qed
	\end{proof}
	
	Note that Lemma~\ref{lem:binary_counting_gadget} provides the first
	part towards the induction step of constructing a valid $H_{k + 1}$ from
	a valid $H_k$. 
	In the following,
	we show how to scale periods of the $a_i$ such that $a_i$ from
	$H_k$ can serve as the $a_{i + 1}$ from $H_{k + 1}$ and $v$ can serve
	as the new $a_1$.  Afterwards, it remains to show two things.  First,
	additional edges to actually build $H_{k + 1}$ from $H_k$ can be
	introduced without losing validity.  And secondly, we need the
	initial step of the induction, i.e., the existence of a valid $H_1$
	even in the presence of only one capacity change.
	
	We say that a minimum dynamic cut $S$ \emph{remains optimal under
		scaling and translation of time} if $S$ is a minimum cut on graph
	$G=(V,E)$ with transit times $\tau_e$, capacities $u_e(\Theta)$ and
	time interval $[0,T]$ if and only if $\hat{S}$ with
	$\hat{S}_v(r\cdot\Theta+T_0)\coloneqq S_v(\Theta)\ \forall\Theta\in[0,T]$ is a
	minimum cut on $\hat{G}=(V,E)$ with transit times
	$\hat{\tau}_e\coloneqq r\cdot\tau_e$, capacities
	$\hat{u}_e(r\cdot\Theta+T_0)\coloneqq u_e(\Theta)$ and time interval
	$[T_0,r\cdot T+T_0]$ for any $r\in\mathbb{R}^+,T_0\in\mathbb{R}$.
	
	\begin{restatable}{lemma}{affineStability}
		\label{lem:affine_stability}
		Any dynamic cut remains optimal under scaling and translation of time. It also remains optimal under scaling of capacities.
	\end{restatable}
	\begin{proof}
		The capacity of a cut is unaffected by translation of time. Scaling time scales the capacity of any cut by the same factor. So the relative difference of the capacity of different cuts is not affected by scaling and translation of time.
		
		Scaling capacities alters the capacity of every cut by the same factor, so the relative difference in capacity of cuts remains unchanged.
		\qed
	\end{proof}
	
	With this, we can combine binary counting gadgets of different sizes to create a large binary counting gadget $H_\ell$ while only requiring a single capacity change.
	
	\begin{restatable}{lemma}{binaryCountingGadgetFinal}
		\label{lem:binary_counting_gadget_final}
		For every $\ell\in\mathbb{N}^+$, there exists a polynomially sized, acyclic dynamic flow network with only one capacity change that contains a valid binary counting gadget $H_\ell$.
	\end{restatable}
	\begin{proof}
		To create the necessary change patterns for the $a_i$ of $H_\ell$, we chain binary counting gadgets of increasing size, beginning with $H_2$ up to $H_\ell$ together creating the graph $G_\ell$ as shown in Figure~\ref{fig:exponentiallyComplexCut_abstract}. The coarse idea is to ensure the behavior of all $a_{k,i}$ by having them mimic the central vertex $v_{k-i}$ of the correct smaller binary counting gadget.
		
		Timewise, the binary counting gadgets $H_k$ are scaled by $\Delta_k\coloneqq2^{-k}$ and translated by $T_{0,k}\coloneqq2+3(1-2\Delta_k)+\Delta_k$.
		So the first change from $S$ to $\bar{S}$ of the $a_{k,i}$ in $H_k$ should happen at time $T_{0,k}$ and the period between two successive changes of $a_{k,i}$ is $p_{k,i}\coloneqq\Delta_k\cdot2^i=\Delta_{k-i}$, which is the period between two successive changes of $v_{k-i}$. 
		To correctly synchronize the different $H_k$, the delay for the mimicking between the binary counting gadgets is set to $\tau_{(v_{k-i},a_{k,i})}=3\cdot(\Delta_{k-i}-2\Delta_k)+\Delta_k$. This is chosen so that $T_{0,k-i}+2\Delta_{k-i}+\tau_{(v_{k-i},a_{k,i})}=T_{0,k}$. We set $\varepsilon\coloneqq\frac{1}{6}$ for the capacity of $(s,v_k)$ in $H_k$.
		
		To ensure that the connecting mimicking gadgets do not exceed the permitted capacities leaving $v_k$ of $H_k$, the capacities of all edges in the $H_k$ are scaled by factor $\lambda_k\coloneqq\frac{1}{5^k}$. In accordance with this scaling and the capacity requirements of Lemma~\ref{lem:mimicking_gadget}, the capacities of the mimicking gadget connecting $v_{k-i}$ to $a_{k,i}$ are $u_{(v_{k-i},a_{k,i})}=\lambda_k\gamma_i$ with $\gamma_i\coloneqq2^i+4\varepsilon$ and $u_{(a_{k,i},t)}=\lambda_k\varepsilon$.
		
		Ensuring the partition changes of $v_0$ can be done with one capacity change as shown in Figure~\ref{fig:inductionStartForExpCutOriginalVersion}, using a path $s,v_0,t$. Transit times for $H_\text{start}$ are $\tau_{(s,v_0)}=1$ and $\tau_{(v_0,t)}=T-3$, capacities are $u_{(s,v_0)}=2$ and $u_{(v_0,t)}(\Theta)=1,\Theta<2$ changing to $u_{(v_0,t)}(\Theta)=3,\Theta\geq 2$. The time horizon is set to $T\coloneqq6$.
		
		$G_\ell$ clearly has polynomial size in regard to $\ell$ and there is only one capacity change.
		The previously shown functionality of the binary counting gadget in Lemma~\ref{lem:binary_counting_gadget} and the mimicking gadget in Lemma~\ref{lem:mimicking_gadget} are the basis for showing, that the contained $H_\ell$ is valid.
		Lemma~\ref{lem:affine_stability} provides that those gadgets' functionality is also given under the shifted and compressed time in which they are used for the construction of $G_\ell$. To prove the correct behavior of the constructed graph, it needs to be shown that the mimicking gadgets adhere to the capacity restrictions established earlier, and that their attachment to smaller binary counting gadgets does not impede the behavior of those counting gadgets.
		
		The correctness of the mimicking gadgets' capacities can easily be seen. The $a_{k,i}$ have no incoming edges outside of the mimicking gadget, so $u_{(a_{k,i},t)}>0$ fulfills the requirement for $(a_{k,i},t)$. There are two additional edges leaving $a_{k,i}$, one to $b_{k,i}$ and one to $v_k$. The combined capacity of all outgoing edges of $a_{k,i}$ is therefore $\lambda_k(\alpha_i+2^{i-1}+\varepsilon)=\lambda_k(2^i+3\varepsilon)<\lambda_k\gamma_i$, so the restriction for $(v_{k-i},a_{k,i})$ holds.
		
		To see that the chaining of binary counting gadgets does not impede the behavior of smaller binary counting gadgets, notice that the capacities of the edges leaving the $v_k$ are chosen to not cross the established threshold:
		\begin{equation*}
		\sum_{i\in\mathbb{N}^+,i\leq\ell-k}\lambda_{k+i}\gamma_i=\lambda_k\sum_{i\in\mathbb{N}^+,i\leq\ell-k}\left(\frac{2}{5}\right)^i+\frac{4}{6\cdot5^i}<\lambda_k(\frac{2}{3}+\frac{1}{6})=\lambda_k(1-\varepsilon)
		\end{equation*}
		Note that the behavior of $v_0$ is also unimpeded by the connections to the binary counting gadgets. This follows from the same argument for $k=0$ as well as the observation, that the changes of $v_0$ are -- ignoring connections to the binary counting gadgets -- always ensured by a capacity difference of at least $\lambda_0$.
		
		We use induction to show that the binary counting gadgets' $a_{k,i}$ change at the required times for any minimum cut $S$. More specifically we show
		\begin{equation*}
		S_{v_k}\xrightarrow{\Theta\equiv T_{0,k}\bmod2\Delta_{k}}\bar{S}_{v_k}\\
		\bar{S}_{v_k}\xrightarrow{\Theta\equiv T_{0,k}+\Delta_k\bmod2\Delta_{k}}S_{v_k}
		\end{equation*}
		for all
		$\Theta\in\{T_{0,k}+2\Delta_{k},T_{0,k}+3\Delta_k,\dots,T_{0,k}+1+\Delta_{k}\}$.
		
		The correct startup behavior of $v_0$ requires two partition changes. For now we ignore the edges from the attachment to the binary counting gadgets, since they do not affect behavior, as argued above. The partition change $S\xrightarrow{2}\bar{S}$ directly follows from the capacity change of $(v_0,t)$ at time $\Theta=2$ increasing the potential contribution of $v_0\in S(2)$. The other partition change $\bar{S}\xrightarrow{3}S$ is a result of the approaching time horizon $T$, which reduces the potential contribution of $v_0\in S(3)$ to zero.
		
		Now assume, for a fixed $k\in\mathbb{N}$, gadgets $H_\text{start}$ to $H_{k-1}$ work correctly, producing the desired changes. Because of the functionality of the mimicking gadget, with the delay $\tau_{(v_{k-i},a_{k,i})}$, by induction $a_{k,i}$ experiences changes $S_{a_{k,i}}\xrightarrow{\Theta}\bar{S}_{a_{k_i}}$ at times
		\begin{equation*}
		\Theta\equiv T_{0,k-i}+\tau_{(v_{k-i},a_{k,i})}\bmod2\Delta_{k-i} \equiv T_{0,k}\bmod2\Delta_{k-i}
		\end{equation*}
		and changes $\bar{S}_{a_{k,i}}\xrightarrow{\Theta}S_{a_{k_i}}$ at times
		\begin{equation*}
		\Theta\equiv T_{0,k-i}+\Delta_{k-i}+\tau_{(v_{k-i},a_{k,i})}\bmod2\Delta_{k-i} \equiv T_{0,k}+\Delta_{k-i}\bmod2\Delta_{k-i}
		\end{equation*}
		beginning with $\Theta=T_{0,k-i}+2\Delta_{k-i}+\tau_{(v_{k-i},a_{k,i})}=T_{0,k}$ up to $\Theta=T_{0,k-i}+1+\Delta_{k-i}+\tau_{(v_{k-i},a_{k,i})}=T_{0,k}+1-\Delta_{k-1}$. This means that beginning at $T_{0,k}$ the $a_{k,i}$ form a binary counter increasing every $2\Delta_k$ with the additional change $\bar{S}_{a_k}\xrightarrow{T_{0,k}+1}S_{a_k}$. Now Lemma~\ref{lem:binary_counting_gadget} -- multiplying the time with factor $\Delta_k$ and adding the initial offset of $T_{0,k}$ -- provides the desired change timings for $v_k$.
		
		The correct change pattern of the $a_{\ell,i}$ required for the validity of $H_\ell$ follow from the stronger induction hypothesis for the $v_k$, as seen above during the induction step.
		\qed
	\end{proof}
	
	\begin{figure}[t]
		\centering
		\includegraphics[width=\textwidth]{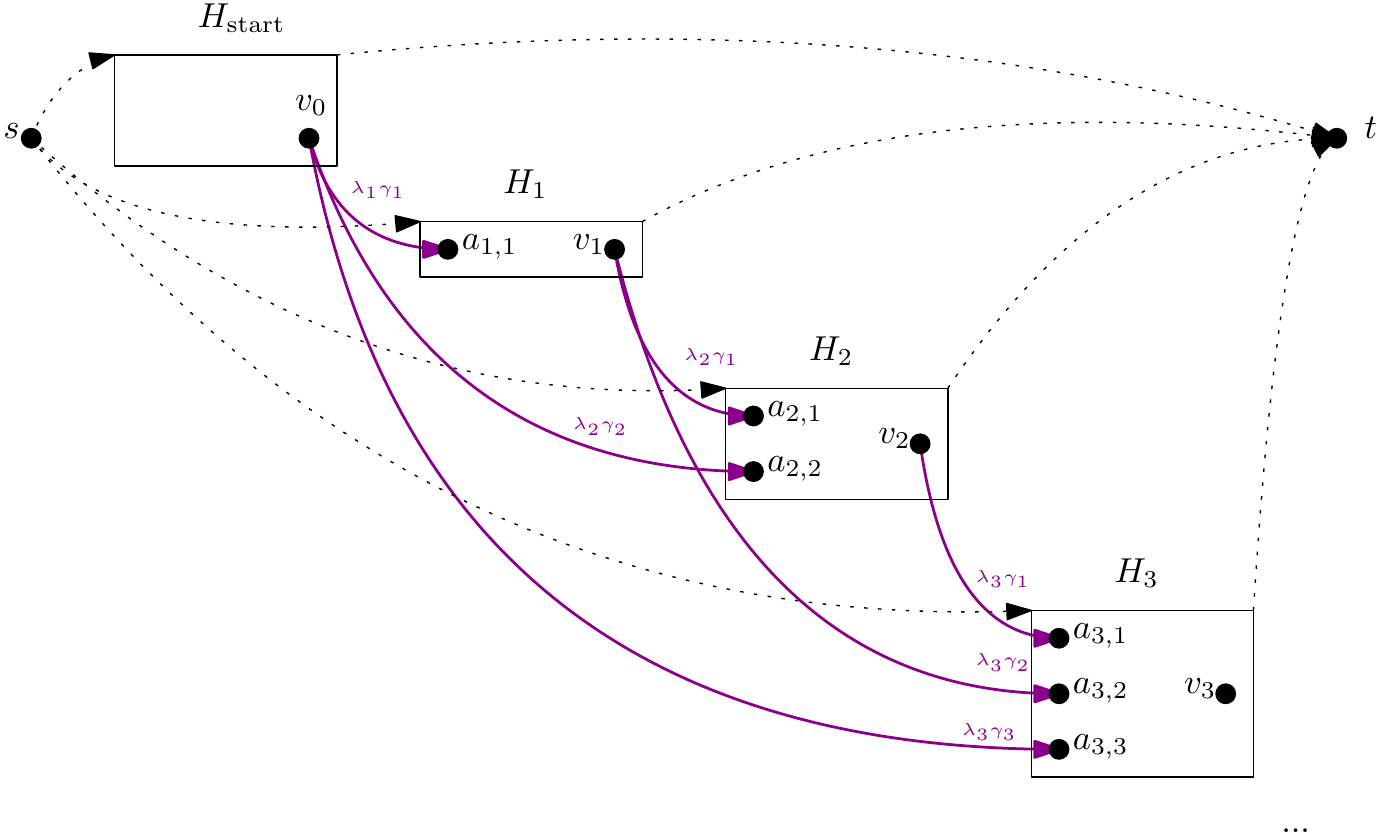}
		\caption{Construction linking binary counting gadgets to ensure $\ch_{v_\ell}=2^\ell$ partition changes at $v_\ell$ in a minimum cut; purple edges represent mimicking gadgets, numbers are capacities.}
		\label{fig:exponentiallyComplexCut_abstract}
	\end{figure}
	
	To be able to use the complexity of minimum cuts to show complexity of maximum flows, we need the following lemma.
	
	\begin{restatable}{lemma}{CutRestrictsFlow}
		\label{obs:CutRestrictsFlow}
		Every edge contributing to the capacity of some minimum cut has to
		be saturated by every maximum flow during the time where it
		contributes to a cut.
		Moreover, every edge $e = (v,w)$ with $v\in\bar{S}(\Theta)$ and
		$w\in S(\Theta+\tau_e)$ for some minimum cut $S$ may not route flow
		at time $\Theta$ for any maximum flow.
	\end{restatable}
	\begin{proof}
		For any minimum cut $S$, flow is routed from $s\in S$ to $t\in\bar{S}$. This means that any path along which flow is routed has to contain at least one edge allowing flow to move from $S$ to $\bar{S}$. All edges allowing flow to traverse from $S$ to $\bar{S}$ contribute to $\capacity(S)$. As such, if one contributing edge was not saturated by a flow $f$, the value $|f|$ would be smaller than $\capacity(S)$. Then the cut-flow duality of Theorem~\ref{trm:cut_flow_duality} provides that $f$ cannot be a maximum flow.
		
		Given a minimum cut $S$, flow of $f$ routed over an edge $e=(v,w)$ with $v\in\bar{S}(\Theta)$ and $w\in S(\Theta+\tau_e)$ has to cross multiple edges from $S$ to $\bar{S}$ to reach $t$, one before $e$ and one after it. So even if $f$ saturates all edges contributing to $\capacity(S)$ as discussed above, less flow than $\capacity(S)$ can reach $t$. With this, Theorem~\ref{trm:cut_flow_duality} again provides that $f$ is no maximum flow.
		\qed
	\end{proof}
	
	To obtain the following theorem, it only remains to observe that the
	structure of the minimum cut in the construction of
	Lemma~\ref{lem:binary_counting_gadget_final} also implies
	exponentially complex maximum flows, using
	Lemma~\ref{obs:CutRestrictsFlow}.  Further note that
	discretization of time is possible.
	
	\begin{restatable}{theorem}{trmExponentiallyComplexCut}
		\label{trm:exponentially_complex_cut}
		There exist dynamic flow networks with only one capacity change
		where every minimum cut and maximum flow has exponential complexity.
		This even holds for acyclic networks and discrete time.
	\end{restatable}
	\begin{proof}
		A valid binary counting gadget $H_\ell$ has a central vertex $v_\ell$ that experiences $2^k$ partition changes in any minimum cut $S$, as shown in Lemma~\ref{lem:binary_counting_gadget}. Since Lemma~\ref{lem:binary_counting_gadget_final} provides the existence of a polynomially sized, acyclic dynamic flow network $G_\ell$ containing a valid $H_\ell$, so any minimum cut in $G_\ell$ has exponential complexity. The construction of $G_\ell$ uses only transit times and change timings that are multiples of $\Delta_\ell$, so discretizing time to units of length $\Delta_\ell$ provides a discrete time dynamic flow network with the same properties.
		
		To see that this partition change pattern with exponentially many changes also implies that any maximum flow has to have exponential complexity, we need Lemma~\ref{obs:CutRestrictsFlow}, which shows that the minimum cuts impose restrictions on maximum flows.
		The partition change pattern of $a_{\ell,1}$, changing every $2\Delta_\ell$ and $v_\ell$, changing every $\Delta$ in $S$ results in $(a_{\ell,1},v_\ell)$ changing from an edge from $S$ to $\bar{S}$ to an edge from $\bar{S}$ to $S$ exponentially often.
		This implies exponentially many changes of $f_{(a_{\ell,1},v_\ell)}$ in any maximum flow $f$ in $G_\ell$.
		\qed
	\end{proof}
	
	\begin{figure}[t]
		\centering
		\includegraphics{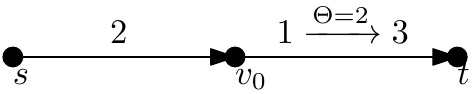}
		\caption{Construction of $H_\text{start}$, providing the partition changes of $v_0$ needed for $G_\ell$ with one capacity change, numbers are capacities, $\tau_{(v_0,t)}=T-3$.}
		\label{fig:inductionStartForExpCutOriginalVersion}
	\end{figure}
	
	Note that the construction from
	Lemma~\ref{lem:binary_counting_gadget_final} requires a specific time
	horizon $T$.  In the case of infinite considered time, two capacity
	changes suffice to obtain the same result.
	
	\begin{restatable}{corollary}{corExponentiallyComplexCutInfTime}
		\label{cor:exponentially_complex_cut_inf_time}
		Theorem~\ref{trm:exponentially_complex_cut} also holds for infinite
		considered time with two capacity changes.
	\end{restatable}
	\begin{proof}
		The starting gadget $H_\text{start}$ shown in Figure~\ref{fig:inductionStartForExpCutOriginalVersion} can be modified to force the two changes $S\xrightarrow{2}\bar{S}$ and $\bar{S}\xrightarrow{3}S$ of $v_0$ with infinite considered time by additionally changing the capacity of $(v_0,t)$ to zero at time $\Theta=3$. The rest of the proof of Theorem~\ref{trm:exponentially_complex_cut} is not changed by the introduction of infinite considered time.
		\qed
	\end{proof}
	
	Note that if Conjecture~\ref{cjt:infTimeSingleCapacityChange_poly} holds, two capacity changes are necessary in this setting.
	
	As mentioned in the introduction, the above complexity results
	transfer to the setting where we have time-dependent transit times
	instead of time-dependent capacities. The result of Corollary~\ref{cor:exponentially_complex_cut_inf_time} can even be strengthened to only require a single transit time change.
	
	\begin{restatable}{corollary}{corExponentiallyComplexTransitTime}
		Theorem~\ref{trm:exponentially_complex_cut} also holds in the
		setting of static capacities and time-dependent transit times with a
		single change, with finite time horizon and with infinite considered
		time.
	\end{restatable}
	\begin{proof}
		For infinite considered time, we give the starting gadget $H_\text{start}$ presented in Figure~\ref{fig:inductionStartForExpCutInfiniteTimeDynamicTransitTimes}. Here $(v_0,t)$ is always saturated for any maximum flow $f$, except during $[2,3]$, when no flow can be routed over it due to the increase in transit time by $1$ of $(s,v_0)$. The corresponding minimum cut $S$ requires $v_0$ to be in $S$ always except during $[2,3]$, when $v_0$ has to be in $\bar{S}$. These are the changes $H_\text{start}$ needs to provide the induction start for the proof of Theorem~\ref{trm:weak_NP_hardness}.
		
		This clearly also works for only considering time in $[0,6]$ as in the proof of Theorem~\ref{trm:weak_NP_hardness}.
		\qed
	\end{proof}
	\begin{figure}
		\centering
		\includegraphics{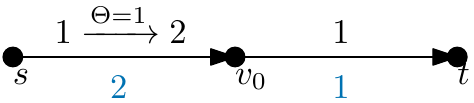}
		\caption{Construction of $H_\text{start}$, providing the partition changes of $v_0$ needed for $G_\ell$ with one change in transit time, black numbers are transit times, blue numbers are capacities.}
		\label{fig:inductionStartForExpCutInfiniteTimeDynamicTransitTimes}
	\end{figure}
	
	Note that the construction of
	Theorem~\ref{trm:exponentially_complex_cut} causes every minimum cut
	and every maximum flow to have exponential complexity.  In the
	following we show that exponentially many changes in cut or flow can
	occur independently.  Specifically, we provide constructions that
	require exponentially complex flows but allow for cuts of low
	complexity and vice versa.
	
	\subsection{Complex Flows and Simple Cuts (and Vice Versa)}
	\label{sec:complex-flows-simple-cuts}
	
	All above constructions require all minimum cuts \emph{and} all
	maximum flows to have exponential complexity.  Here, we show that
	flows and cuts can be independent in the sense that their required
	complexity can be exponentially far apart (in both directions).
	
	\begin{theorem}
		\label{trm:expFlow_simpleCut}
		There exist acyclic dynamic flow networks with only one capacity
		change where every maximum flow has exponential complexity, while
		there exists a minimum cut of constant complexity.  The same is true
		for static capacities and time-dependent transit times.
	\end{theorem}
	\begin{proof}
		Figure~\ref{fig:exponentialFlow_simpleCut} shows a graph with these
		properties. This is achieved by only allowing flow to enter $v_0$
		during $[0,1]$, but it has to leave $v_k$ during $[0,2^k]$ due to
		the reduced capacity of $(v_k,t)$ for a time horizon
		$T\geq2^k$. Apart from $v_k$ all $v_i$ are connected to the next
		$v_{i+1}$ with a pair of edges, with transit times $2^{k-i-1}$ and
		zero, all these edges have capacity $1$. So all $2^k$ paths of
		different transit time through the $v_i$ have to be used to route flow for
		a maximum flow. Every second of those paths has an even transit time, so
		flow has to traverse the edge with transit time zero between $v_{k-1}$ and
		$v_k$ every second integer time interval, which results in
		exponentially many changes in flow over that edge. However assigning
		all $v_i$ to $\bar{S}$ for all time is a minimum cut without
		partition changes.
		This generalizes to time-dependent transit times, as we can block the edge $(s,v_0)$ at time $\Theta=1$ by increasing its transit time to $T$ at that time.
		\qed
	\end{proof}
	\begin{figure}
		\centering
		\includegraphics{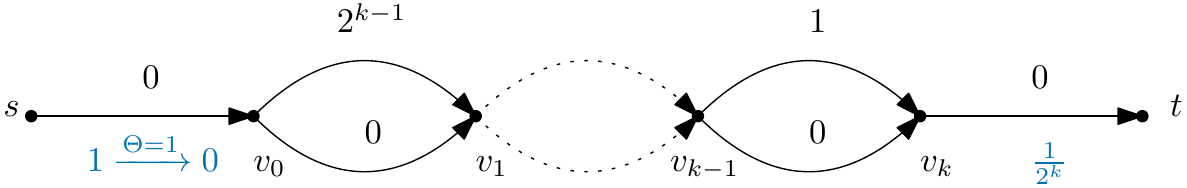}
		\caption{Example of a graph where any maximum flow contains exponentially many changes, but there is a minimum cut with no changes, black numbers are transit times, blue numbers are capacities, unspecified capacities are $1$.}
		\label{fig:exponentialFlow_simpleCut}
	\end{figure}
	
	
	\begin{restatable}{theorem}{thmExpCutSimpleFlow}
		\label{trm:expCut_simpleFlow}
		There exist acyclic dynamic flow networks with only one capacity change
		where every minimum cut has exponential complexity, while there
		exists a maximum flow of linear complexity.  The same is true
		for static capacities and time-dependent transit times.
	\end{restatable}
	\begin{proof}
		Figure~\ref{fig:simpleFlow_exponentialCut} shows a graph where any minimum cut needs to contain exponentially many changes, but a maximum flow with linearly many changes exists. The idea of this construction is that only $\frac{1}{2^k+1}$ flow can enter $s'$ at any time and flow can only leave from $x_{2k}$ to $t$ during one integer interval after $2^k$ and before the time horizon $T=2^k+1$. All other edges have capacity $1$. There is a set of bypass paths through the $y_i$, that allows flow to be routed from $s'$ to $x_{2k}$ in any integer time up to $2^k-1$, by connecting $y_i$ to $y_{i+1}$ with a pair of edges with transit times $2^{k-i-1}$ and zero. So $\frac{2^k}{2^k+1}$ flow can be routed through the graph. The section where exponential cuts will be necessary consists of the $x_i$. The initial $x_0$ can be reached from $s'$ with transit time 1, internally each $x_i,i<k$ is connected with the next $x_i+1$ with a pair of edges with transit times $2^{i+1}$ and zero. The later $x_i,i\geq k$ are connected to the next $x_{i+1}$ with a pair of edges of transit time $2^{2k-i}$ and zero.
		
		If the bypass paths are used no additional flow can move through the upper paths via the $x_i$ because of the capacity of $(s,s')$ and the lack of a transit time $2^k$ path through the $x_i$. There is no maximum flow that saturates any edge except for $(s,s')$ at any time, so, because of the cut flow duality of Theorem~\ref{trm:cut_flow_duality}, this is the only edge that can contribute to the capacity of a minimum cut. Since there are integer transit time paths from $s'$ to $x_{2k}$ for any transit time up to $2^k-1$, we know that $s'$ has to be in $\bar{S}$ during $[1,2^k+1)$ to prevent any other edges from contributing to the cut. This already results in the capacity of the cut being at least $\frac{2^k}{2^k+1}$, so $s'$ has to be in $S$ during $[0,1)$. To prevent any other contributions to the capacity of the cut, $x_k$ needs to change partition exponentially often. Observe that any path's transit time from any $s'$ to $x_k$ plus the $1$ from $(s',x_0)$ marks a timing where $x_j$ has to be in the $S$ partition. Likewise $2^k$ minus path transit times of paths from $x_k$ to $x_{2k}$ mark timings where $x_k$ has to be in the $\bar{S}$ partition. So during any integer interval before $2^k$ starting with an odd time, $x_k$ has to be in $S$ and at any such interval starting at an even time, it has to be in $\bar{S}$.
		
		The flow through this graph can easily be represented in linear time with at most one change in flow rate per edge. Using the bypass paths as described above, no flow gets routed through any $x_i$. The bypass edges are ordered to ensure that flow moves through any $y_i$ during $[2^k-2^{k-i}+1,2^k+1)$ when using all different path transit times as required for this maximum flow. This means that each edge from $y_i$ to $y_{i+1}$ in the bypass edges sends $\frac{2^i}{2^k+1}$ during $[2^k-2^{k-i}+1,2^k-2^{k-i-1}+1)$ for the $2^{k-i-1}$ transit time edge and during $[2^k-2^{k-i-1}+1,2^k+1)$ for the edge with transit time zero. Flow over the remaining edges is easily representable as well; $(s,s')$ is saturated during $[1,T)$, $(x_{2k},t)$ sends $\frac{2^k}{2^k+1}$ during $[2^k,T)$ and $(s',y_0),(y_k,x_{2k})$ only exist to improve the visual representation, flow over them is given by $(s,s'),(x_{2k},t)$.
		
		This generalizes to time-dependent transit times, as we can activate the edge $(x_{2k},t)$ at time $\Theta=2^k$ by decreasing its transit time from $T$ to $0$ at that time.
		\qed
	\end{proof}
	\begin{figure}
		\centering
		\includegraphics[width=\textwidth]{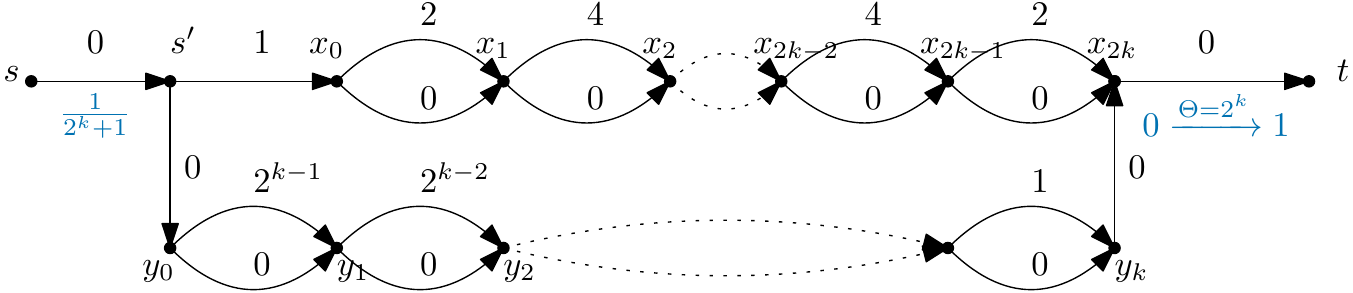}
		\caption{Example of a graph where any minimum cut contains exponentially many changes, but there is a flow with only linearly many changes, black numbers are transit times, blue numbers are capacities, unspecified capacities are $1$, time horizon is $T=2^k+1$.}
		\label{fig:simpleFlow_exponentialCut}
	\end{figure}
	
	\newpage
	\bibliography{references}

	\newpage
	\appendix
	
	\section{On the Strong NP-Completeness of Dynamic Flows with
		Time-Dependent Capacities}
	\label{sec:strong-np-compl}
	
	As mentioned in the introduction, a poof for strong NP-completeness
	for the dynamic flow problem with time-dependent capacities was
	claimed by Sha Cai and Wong \cite[Theorem~2]{sha2000maximum}. The
	exponential complexity we proof in
	Theorem~\ref{trm:exponentially_complex_cut} does not disproof that the
	maximum dynamic flow problem is in NP, but it invalidates the use of
	the two canonical witnesses for verifying a solution in polynomial
	time. As such the claim that the maximum dynamic flow problem is
	obviously in NP is in doubt.
	
	The reduction from the 3-Dimensional Matching Problem suffers from
	the issue that flow can take a path from $s$ to $t$ using edges
	belonging to different triplets of te 3DM. One example of this is
	visualized in
	Figure~\ref{fig:dynamicFlow_3DM_reduction_counterexample}. The 3DM
	Problem has four possible triples
	$M=\{(w_1,x_1,y_1),(w_1,x_2,y_2),(w_2,x_2,y_3),(w_3,x_3,y_3)\}$ to hit
	each of the three elements of each set
	$\{w_1,w_2,w_3\},\{x_1,x_2,x_3\},\{y_1,y_2,y_3\}$ exactly once. This
	is impossible because the need to hit $y_1$ and $w_3$ necessitates the
	use of $(w_1,x_1,y_1),(w_3,x_3,y_3)$, but the remaining three elements
	$w_2,x_2,y_2$ cannot be covered by either of the remaining triples in
	$M$. However there are three edge disjoint paths
	$(w_1,x_1,y_1),(w_2,x_2,y_2),(w_3,x_3,y_3)$ in the induced graph,
	allowing a flow of $3$ to pass from $s$ to $t$, which should
	correspond to a solvable 3DM instance.
	
	\begin{figure}
		\centering
		\includegraphics{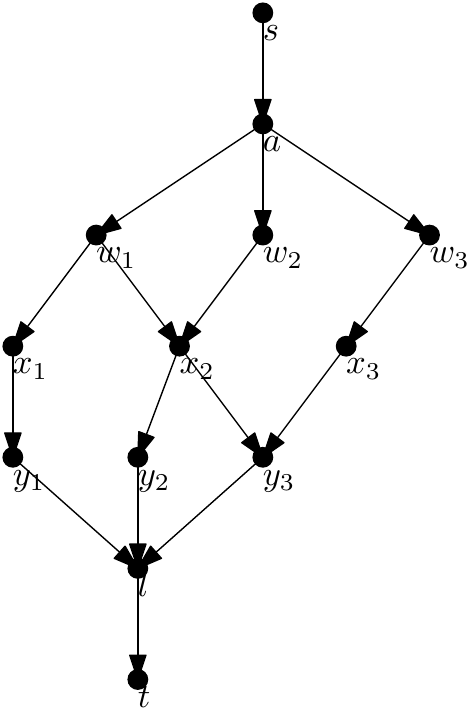}
		\caption{Visualization of the counterexample to the reduction of 3DM to maximum dynamic flow, where the 3DM instance is not solvable, but the induced maximum dynamic flow instance achieves the required flow of $3$, implying a solution to the 3DM problem.}
		\label{fig:dynamicFlow_3DM_reduction_counterexample}
	\end{figure}
	
	Furthermore, it seems unlikely that the dynamic flow problem with time-dependent capacities, where capacities are piecewise constant functions, capacity changes happen at integer times and edge transit times are integer, is strongly NP hard. This case can be solved in pseudo-polynomial time using the time expanded graph, which has one vertex for every integer time and edges connecting instances with the correct time difference. In the time expanded graph, there are no more changing capacities, so it can be solved using temporally repeated flows. With this in mind a strong NP-hardness proof would show $\text{P}=\text{NP}$.
	
	\section{Strong NP-Hardness for Simple Flow Paths}
	\label{sec:strong-np-hardness}
	
	During our studies, we stumbled upon the following related NP-hardness
	reduction.  However, it is somewhat beyond the scope of the paper and
	thus only mentioned here in the appendix.
	
	\begin{theorem}
		\label{trm:strong_NP_hardness_simple_paths}
		The (maximum) dynamic flow problem restricted to simple flow paths with time-dependent capacities or time-dependent transit times is strongly NP-hard.
	\end{theorem}
	\begin{proof}
		This follows from the NP-hardness proof for the multi-commodity flow over time problem with simple flow paths and without storage presented by Hall, Hippler and Skutella \cite[Theorem~7]{hall2007multicommodity}. Their construction requires a traffic light gadget, but otherwise works with only a single commodity of flow. The traffic light gadget allows setting the capacity of an edge to zero except for a interval $[a,b)$ during which the edge offers usable capacity. Since the scenario considered here allows adjusting an edge's capacity based on time, creation of such a gadget is trivial for time-dependent capacities. For time-dependent transit times, we can construct a traffic light gadget, allowing flow over $e$ only during $[a,b)$ by setting the transit time of $e$ to $T$ during $[0,a)$ and $[b,T)$.
		\qed
	\end{proof}
	
\end{document}